\documentclass[a4paper,UKenglish,cleveref, autoref, thm-restate]{lipics-v2021}
\usepackage{combelow}
\usepackage{enumerate}
\usepackage{hyperref}

\DeclareMathOperator{\aand}{\,\;\wedge\;\,}

\newcommand\inde[1]{ \hskip 0.3mm plus 0.2mm minus 0.1mm [ \hskip 0.2mm plus 0.2mm minus 0.1mm {#1} \hskip 0.2mm plus 0.2mm minus 0.1mm ] \hskip 0.2mm plus 0.1mm minus 0.1mm}

\title{Closure Properties of General Grammars\,---\allowbreak\,Formally Verified}

\author{Martin Dvorak}{Institute of Science and Technology Austria, Austria}{martin.dvorak@ista.ac.at}{https://orcid.org/0000-0001-5293-214X}{}
\author{Jasmin Blanchette}{Ludwig-Maximilians-Universität München, Germany}{jasmin.blanchette@ifi.lmu.de}{https://orcid.org/0000-0002-8367-0936}{%
This research has received funding from the Netherlands Organization for
Scientific Research (NWO) under the Vidi program (project No.\ 016.Vidi.189.037,
Lean Forward).}

\authorrunning{M.~Dvorak and J.~Blanchette}

\Copyright{Martin Dvorak and Jasmin Blanchette}



\category{}

\relatedversion{}

\supplement{https://github.com/madvorak/grammars/tree/publish}

\acknowledgements{We thank Vladimir Kolmogorov for making this collaboration possible.
We thank Václav Kon\v{c}ický for discussing ideas about the Kleene star construction.
We thank Patrick Johnson, Floris van Doorn, and Damiano Testa for their small yet very valuable contributions to our code.
We thank Eric Wieser for simplifying one of our proofs.
We thank Mark Summerfield for suggesting textual improvements.
We thank the anonymous reviewers for very helpful comments.
Finally, we thank the Lean community for helping us with various technical issues and answering many questions.}

\nolinenumbers

\EventEditors{Adam Naumowicz and Ren\'{e} Thiemann}
\EventNoEds{2}
\EventLongTitle{14th International Conference on Interactive Theorem Proving (ITP 2023)}
\EventShortTitle{ITP 2023}
\EventAcronym{ITP}
\EventYear{2023}
\EventDate{July 31 to August 4, 2023}
\EventLocation{Bia\l{}ystok, Poland}
\EventLogo{}
\SeriesVolume{268}
\ArticleNo{15}

\let\oldtexttt\texttt
\renewcommand{\texttt}[1]{\begingroup\setlength{\fboxsep}{0pt}\setlength{\fboxrule}{0pt}\colorbox{lipicsLightGray}{\rule[-.25\baselineskip]{0pt}{0.9\baselineskip}\oldtexttt{#1}}\endgroup}

\begin{document}

\maketitle

\begin{abstract}
We formalized general (i.e., type-0) grammars using the Lean 3 proof
assistant. We defined basic notions of rewrite rules and of words derived by a
grammar, and used grammars to show closure of the class of type-0
languages under four operations:\ union, reversal, concatenation, and
the Kleene star.
The literature mostly focuses on Turing machine arguments, which
are possibly more difficult to formalize. For the Kleene star, we could
not follow the literature and came up with our own grammar-based construction.
\end{abstract}

\section{Introduction}
\label{sec:introduction}

The notion of formal languages lies at the heart of computer science.
There are several formalisms that recognize formal languages,
including Turing machines and formal grammars.
In particular, both Turing machines and general grammars (also called
type-0 grammars or unrestricted grammars) characterize the same class
of languages, namely, the recursively enumerable or type-0 languages.

There has been work on formalizing Turing machines in proof assistants
\cite{mizarTM, matitaTM, isabelleTM, CarneiroPRF, KunzeTuring,
Cook_Levin-AFP}.
General grammars are an interesting alternative because
they are easier to define than Turing machines, and some proofs about
general grammars are much easier than the proofs of similar properties
of Turing machines.

We therefore chose general grammars as
the basis for our Lean 3 \cite{Lean} library of results about
recursively enumerable or type-0 languages. \pagebreak[2]
The definition of grammars \cite{Aho} consists of several layers of concepts
(Section~\ref{sec:definitions}):
\begin{itemize}
\item the type of symbols is the disjoint union of terminals and nonterminals;
\item rewrite rules are pairs of the form $u \rightarrow v$, where $u$ and
  $v$ are strings over symbols and
  $u$~contains at least one nonterminal;
\item a grammar is a tuple consisting of a type of terminals, a type of
  nonterminals, an initial symbol $S$, and a set of rewrite rules;
\item application of a rewrite rule $u \rightarrow v$ to a string $\alpha u
  \beta$ is written $\alpha u \beta \Rightarrow \alpha v \beta$;
\item the derivation relation $\Rightarrow^*$ is the reflexive transitive
  closure of the $\Rightarrow$ relation;
\item a grammar derives a word $w$ if $S \Rightarrow^* w$;
\item the language generated by a grammar is the set of words derived by it;
\item a language is type 0 if there exists a grammar that generates it.
\end{itemize}

\noindent We formalized four closure properties of type-0 languages.

The first
such property we present is closure of type-0 languages under
union (Section~\ref{sec:closure-under-union}).
We followed the standard construction for context-free
grammars, which incidentally works for general grammars as well.

The second closure property we formalized is closure under reversal
(Section~\ref{sec:closure-under-reversal}). This was straightforward.

The third closure property we formalized is closure under concatenation
(Section~\ref{sec:closure-under-concatenation}). The main difficulty was to
avoid matching strings on the boundary of the concatenation. This issue does not
arise with context-free grammars because only single symbols are matched
and these are tidily located on either side of the boundary.

The fourth and last closure property we formalized is closure under the Kleene star
(Section~\ref{sec:closure-under-kleene-star}). This was the most difficult part
of our work. Because the literature mostly focuses on Turing machine arguments,
we needed to invent our own construction. We first developed a
detailed proof sketch and then formalized it. The sketch is
included in this paper.

One closure property we did not formalize is closure under intersection. The
reason is that we are not aware of any elegant construction based on grammars
only.
Recall that type-0 languages are not closed under complement, as witnessed by
the halting problem \cite{Ullman}.

Our development is freely available online.\footnote{\label{ourl}
\url{https://github.com/madvorak/grammars/tree/publish}}
It consists of about 12\,500 lines of spaciously formatted Lean code.
It uses the Lean 3 mathematical library \texttt{mathlib}
\cite{mathlib-mathlib-2020}.\footnote{\label{mathlib-url}
\url{https://github.com/leanprover-community/mathlib/tree/7ed4f2cec2}}
We also benefited from the metaprogramming framework
\cite{Lean3metaprogramming}, which allowed us to easily develop small-scale
automation that helped make some proofs less verbose.

Although Lean is based on dependent type theory \cite{LuoTT}, our code uses
only nondependent types for data. We still found dependent type theory useful
for bound-checked indexing of lists using the function
\href{https://github.com/leanprover-community/lean/blob/154ac72f4ff674bc4486ac611f926a3d6b999f9f/library/init/data/list/basic.lean#L90}
{\texttt{list.nth\_le}}
(which takes a list, an index, and a proof that the index is within bounds as arguments).
We did not attempt to make our development constructive.

\section{Definitions}
\label{sec:definitions}

\subsection{Grammars}

As outlined in the introduction, the definition of grammars consists of
several layers of declarations.

\lstset{mathescape}

Symbols are essentially defined as a sum type of terminals \texttt{T} and nonterminals \texttt{N}.
However, we want to refer to terminals and nonterminals by name (using
\href{https://github.com/madvorak/grammars/blob/99d9f3e6e2fca51b4ab146cd332ff881ff937a29/src/classes/general/basics/definition.lean#L10}
{\texttt{symbol.terminal}} and
\href{https://github.com/madvorak/grammars/blob/99d9f3e6e2fca51b4ab146cd332ff881ff937a29/src/classes/general/basics/definition.lean#L11}
{\texttt{symbol.nonterminal}}
instead of
\href{https://github.com/leanprover-community/lean/blob/154ac72f4ff674bc4486ac611f926a3d6b999f9f/library/init/core.lean#L205}{\texttt{sum.inl}} and
\href{https://github.com/leanprover-community/lean/blob/154ac72f4ff674bc4486ac611f926a3d6b999f9f/library/init/core.lean#L206}{\texttt{sum.inr}}), so
we define symbols as an inductive type:
\begin{lstlisting}
inductive symbol (T : Type) (N : Type)
| terminal    : T $\rightarrow$ symbol
| nonterminal : N $\rightarrow$ symbol
\end{lstlisting} \pagebreak[2]
We do not require \texttt{T} and \texttt{N} to be finite. As a result, we do not need to copy the typeclass instances \texttt{[fintype T]} and \texttt{[fintype N]} alongside our type parameters (which would appear in almost every lemma
statement). Instead, later we work in terms of a list of rewrite rules, which is finite by
definition and from which we could infer that only a finite set of terminals and
a finite set of nonterminals can occur.

The left-hand side $u$ of a rewrite rule $u \rightarrow v$ consists of three
parts (an arbitrary string $\alpha$, a nonterminal $A$, and another arbitrary
string $\beta$, such that $u = \alpha A \beta$):
\begin{lstlisting}
structure grule (T : Type) (N : Type) :=
(input_L : list (symbol T N))
(input_N : N)
(input_R : list (symbol T N))
(output_string : list (symbol T N))
\end{lstlisting}
An advantage of this representation is that we do not need to carry the
proposition ``the left-hand side contains a nonterminal'' around.
A disadvantage is that we subsequently need to concatenate more terms.

A definition of a general grammar follows. Notice that only the type argument
\texttt{T} is part of its type:
\begin{lstlisting}
structure grammar (T : Type) :=
(nt : Type)
(initial : nt)
(rules : list (grule T nt))
\end{lstlisting}
Later we can use the dot notation to access individual fields.
For example, if \texttt{g} is a term of the type \texttt{grammar T}, we can write \texttt{g.nt} to access the type of its nonterminals. By writing \texttt{(g.rules.nth\_le 0 \_).output\_string} we
obtain the right-hand side of the first rewrite rule in~\texttt{g}. The underscore, when not inferred automatically, must be replaced by a term of the type \texttt{0 < g.rules.length}, which is a proof that the list \texttt{g.rules} is not empty.

The next line adds an implicit type argument \texttt{T} to all declarations that come after:
\begin{lstlisting}
variables {T : Type}
\end{lstlisting}

The following definition captures the application $\Rightarrow$ of a rewrite rule:
\begin{lstlisting}
def grammar_transforms (g : grammar T)
    (w$_1$ w$_2$ : list (symbol T g.nt)) :
  Prop :=
$\exists$ r : grule T g.nt,
  r $\in$ g.rules $\ \ \ \wedge$
  $\exists$ u v : list (symbol T g.nt),
    w$_1$ = u ++ r.input_L ++ [symbol.nonterminal r.input_N]
          $\;\,$++ r.input_R ++ v $\ \ \ \wedge$
    w$_2$ = u ++ r.output_string ++ v
\end{lstlisting}
The operator
\href{https://github.com/leanprover-community/lean/blob/154ac72f4ff674bc4486ac611f926a3d6b999f9f/library/init/data/list/basic.lean#L36}
{\texttt{++}}
concatenates two lists. We can view
\href{https://github.com/madvorak/grammars/blob/99d9f3e6e2fca51b4ab146cd332ff881ff937a29/src/classes/general/basics/definition.lean#L30}
{\texttt{grammar\_transforms}}
as a function that takes a grammar \texttt{g} over the terminal type \texttt{T} and outputs a binary relation over strings of the type that \texttt{g} works internally with.

The part \texttt{r.input\_L ++ [symbol.nonterminal r.input\_N] ++ r.input\_R} represents
the left-hand side of the rewrite rule \texttt{r}. Note that the terms
\texttt{r.input\_L} and \texttt{r.input\_N} cannot be concatenated
directly, since they have different types. The
term \texttt{r.input\_N} must first be wrapped in
\href{https://github.com/madvorak/grammars/blob/99d9f3e6e2fca51b4ab146cd332ff881ff937a29/src/classes/general/basics/definition.lean#L11}
{\texttt{symbol.nonterminal}}
to go from the type \texttt{g.nt} to the type \hbox{\texttt{symbol T g.nt}} and then surrounded by \texttt{[\,]} to become a (singleton) list.

The derivation relation $\Rightarrow^*$ is defined from
\href{https://github.com/madvorak/grammars/blob/99d9f3e6e2fca51b4ab146cd332ff881ff937a29/src/classes/general/basics/definition.lean#L30}
{\texttt{grammar\_transforms}} using the reflexive transitive closure:
\begin{lstlisting}
def grammar_derives (g : grammar T) :
  list (symbol T g.nt) $\rightarrow$ list (symbol T g.nt) $\rightarrow$ Prop :=
relation.refl_trans_gen (grammar_transforms g)
\end{lstlisting}
Consequently, proofs about derivations will use structural induction.

The predicate ``to be a word generated by the grammar \texttt{g}'' is defined as
the special case of the relation \texttt{grammar\_derives g} where the left-hand
side is fixed to be the singleton list made of the initial symbol of \texttt{g} and
the right-hand side is required to consist of terminal symbols only:
\begin{lstlisting}
def grammar_generates (g : grammar T) (w : list T) : Prop :=
grammar_derives g [symbol.nonterminal g.initial]
  (list.map symbol.terminal w)
\end{lstlisting}

\subsection{Languages}

In our entire project, we work with the following definition of languages provided
by \texttt{mathlib} in the
\href{https://github.com/leanprover-community/mathlib/blob/7ed4f2cec258c1a1f86cd9a45adf7dc335f42ee1/src/computability/language.lean#L24}
{\texttt{computability}} package:
\begin{lstlisting}
def language ($\alpha$ : Type*) := set (list $\alpha$)
\end{lstlisting}
The type argument $\alpha$ is instantiated by our terminal type \texttt{T} in all places
where we work with languages.
We do not mind restricting \texttt{T} to be \texttt{Type} since we are not interested in languages
over types from \texttt{Type 1} and higher universes.

The language of the grammar \texttt{g} is defined as the set of all \texttt{w} that satisfy the predicate \texttt{grammar\_generates g w} declared above:
\begin{lstlisting}
def grammar_language (g : grammar T) : language T :=
set_of (grammar_generates g)
\end{lstlisting}
Note that the type parameter \texttt{T} is preserved, but \texttt{g.nt} does not matter in the description of what words are generated. It corresponds to our intuition that the type of terminals is a part of the interface, but the type of nonterminals is an implementation matter.

This is the first time that our custom types meet the standard \texttt{mathlib} type
\href{https://github.com/leanprover-community/mathlib/blob/7ed4f2cec258c1a1f86cd9a45adf7dc335f42ee1/src/computability/language.lean#L24}{\texttt{language}},
which is already connected to many useful types, such as the type of regular expressions.

Finally, we define the class of type-0 languages:
\begin{lstlisting}
def is_T0 (L : language T) : Prop :=
$\exists$ g : grammar T, grammar_language g = L
\end{lstlisting}
All top-level theorems about type-0 languages are expressed in terms of the
\href{https://github.com/madvorak/grammars/blob/99d9f3e6e2fca51b4ab146cd332ff881ff937a29/src/classes/general/basics/definition.lean#L50}
{\texttt{is\_T0}} predicate.

Note that the type system distinguishes between a list of terminals and a list of symbols
that happen to be terminals.
Languages are defined as sets of the former, whereas derivations
in the grammar work with the latter. \pagebreak[2]

In a similar way, we define
\href{https://github.com/madvorak/grammars/blob/99d9f3e6e2fca51b4ab146cd332ff881ff937a29/src/classes/context_free/basics/definition.lean#L6}
{\texttt{CF\_grammar}},
\href{https://github.com/madvorak/grammars/blob/99d9f3e6e2fca51b4ab146cd332ff881ff937a29/src/classes/context_free/basics/definition.lean#L15}
{\texttt{CF\_transforms}},
\href{https://github.com/madvorak/grammars/blob/99d9f3e6e2fca51b4ab146cd332ff881ff937a29/src/classes/context_free/basics/definition.lean#L21}
{\texttt{CF\_derives}},
\href{https://github.com/madvorak/grammars/blob/99d9f3e6e2fca51b4ab146cd332ff881ff937a29/src/classes/context_free/basics/definition.lean#L25}
{\texttt{CF\_generates}},
\href{https://github.com/madvorak/grammars/blob/99d9f3e6e2fca51b4ab146cd332ff881ff937a29/src/classes/context_free/basics/definition.lean#L29}
{\texttt{CF\_language}}, and the
\href{https://github.com/madvorak/grammars/blob/99d9f3e6e2fca51b4ab146cd332ff881ff937a29/src/classes/context_free/basics/definition.lean#L33}
{\texttt{is\_CF}}
predicate for the formal definition of context-free languages.

The theorem
\href{https://github.com/madvorak/grammars/blob/99d9f3e6e2fca51b4ab146cd332ff881ff937a29/src/classes/context_free/basics/inclusion.lean#L74}
{\texttt{CF\_subclass\_T0}}
connects the context-free languages to the type-0 languages.
Type-0 languages remain the main focus of our work. \pagebreak[2]

\subsection{Operations}

The operations under which we prove closure are defined below.

Union is defined in \texttt{mathlib} as follows:
\begin{lstlisting}
protected def set.union (s$_1$ s$_2$ : set $\alpha$) : set $\alpha$ :=
{a | a $\in$ s$_1$ $\vee$ a $\in$ s$_2$}
\end{lstlisting}
The following declaration in \texttt{mathlib} states that the union of languages is denoted by writing the \texttt{+} operator between two terms of the
\href{https://github.com/leanprover-community/mathlib/blob/7ed4f2cec258c1a1f86cd9a45adf7dc335f42ee1/src/computability/language.lean#L24}{\texttt{language}} type:
\begin{lstlisting}
instance : language.has_add (language $\alpha$) := $\langle$set.union$\rangle$
\end{lstlisting}

We define the reversal of a language as follows:
\begin{lstlisting}
def reverse_lang (L : language T) : language T :=
$\lambda$ w : list T, w.reverse $\in$ L
\end{lstlisting}
We do not declare any syntactic sugar for reversal.

Concatenation is defined using the following general \texttt{mathlib} definition:
\begin{lstlisting}
def set.image2 (f : $\alpha$ $\rightarrow$ $\beta$ $\rightarrow$ $\gamma$) (s : set $\alpha$) (t : set $\beta$) : set $\gamma$ :=
{c | $\exists$ a b, a $\in$ s $\wedge$ b $\in$ t $\wedge$ f a b = c}
\end{lstlisting}
The next \texttt{mathlib} declaration states that concatenation of languages is denoted by writing the \texttt{*} operator between two terms of the
\href{https://github.com/leanprover-community/mathlib/blob/7ed4f2cec258c1a1f86cd9a45adf7dc335f42ee1/src/computability/language.lean#L24}{\texttt{language}} type:
\begin{lstlisting}
instance : language.has_mul (language $\alpha$) := $\langle$set.image2 (++)$\rangle$
\end{lstlisting}

The Kleene star of a language is defined in \texttt{mathlib} as follows:
\begin{lstlisting}
def language.star (l : language $\alpha$) : language $\alpha$ :=
{x | $\exists$ S : list (list $\alpha$), x = S.join $\wedge$ $\forall$ y $\in$ S, y $\in$ l}
\end{lstlisting}
We do not declare any syntactic sugar for the Kleene star.

\section{Closure under Union}
\label{sec:closure-under-union}

In this section, we prove the following theorem:
\begin{lstlisting}
theorem T0_of_T0_u_T0 (L$_1$ : language T) (L$_2$ : language T) :
  is_T0 L$_1$ $\wedge$ is_T0 L$_2\;\;\,\rightarrow\;\;\;$is_T0 (L$_1$ + L$_2$)
\end{lstlisting}

The proof of closure of type-0 languages under union consists of three main ingredients:
\begin{enumerate}
\item [\,(1)\!] a construction of a new grammar \texttt{g} from any two given grammars \texttt{g$_1$} and \texttt{g$_2$};
\item [\,(2)\!] a proof that any word generated by \texttt{g$_1$} or \texttt{g$_2$} can also be generated by \texttt{g};
\item [\,(3)\!] a proof that any word generated by \texttt{g} can be equally generated by \texttt{g$_1$} or \texttt{g$_2$}.
\end{enumerate}
Proofs of the other closure properties are organized analogously.
We describe the proof of closure under union in more detail; it allows us to outline
the main ideas of proving closure properties formally in a simple setting.
Since (3) is usually much more difficult than (2),
we refer to (2) as the ``easy direction'' and to (3) as the ``hard direction''. \pagebreak[2]

The proof of the closure of type-0 languages under union follows the standard
construction, which usually states only (1) explicitly, and leaves (2) and (3) to the reader.
We begin (1) by defining a new type of nonterminals. The nonterminals of \texttt{g} consist of
\begin{itemize}
\item the nonterminals of \texttt{g$_1$} including a mark indicating their origin;
\item the nonterminals of \texttt{g$_2$} including a mark indicating their origin;
\item one new distinguished nonterminal.
\end{itemize}
The Lean type \texttt{option (g$_1$.nt $\oplus$ g$_2$.nt)} encodes this disjoint union.
If \texttt{m} is a nonterminal of type \texttt{g$_1$.nt}, its corresponding nonterminal of type \texttt{g.nt} is \texttt{some (sum.inl~m)}.
If \texttt{n} is a nonterminal of type \texttt{g$_2$.nt}, its corresponding nonterminal of type \texttt{g.nt} is \texttt{some (sum.inr~n)}.
The new distinguished nonterminal is called \texttt{none} and becomes the initial symbol of \texttt{g}.
The rewrite rules of \texttt{g} consist of
\begin{itemize}
\item the rewrite rules of \texttt{g$_1$} with all nonterminals mapped to the larger nonterminal type;
\item the rewrite rules of \texttt{g$_2$} with all nonterminals mapped to the larger nonterminal type;
\item two additional rules that rewrite the initial symbol of \texttt{g} to the initial symbol of \texttt{g$_1$} or \texttt{g$_2$}.
\end{itemize}

To reduce the amount of repeated code in the proof, we developed lemmas that allow us to ``lift'' a grammar with a certain type of nonterminals to a grammar with a larger type of nonterminals while preserving what the grammar derives. Under certain conditions, we can also ``sink'' the larger grammar to the original grammar and preserve its derivations.

These lemmas operate on a structure called
\href{https://github.com/madvorak/grammars/blob/99d9f3e6e2fca51b4ab146cd332ff881ff937a29/src/classes/general/basics/lifting.lean#L35}
{\texttt{lifted\_grammar}}
that consists of the following fields:
\begin{itemize}
\item a smaller grammar \texttt{g$_0$} that represents either \texttt{g$_1$} or \texttt{g$_2$} in case of the proof for union;
\item a larger grammar \texttt{g} with the same type of terminals;
\item a function \texttt{lift\_nt} from \texttt{g$_0$.nt} to \texttt{g.nt};
\item a partial function \texttt{sink\_nt} from \texttt{g.nt} to \texttt{g$_0$.nt};
\item a proposition \texttt{lift\_inj} that guarantees that \texttt{lift\_nt} is injective;
\item a proposition \texttt{sink\_inj} that guarantees that \texttt{sink\_nt} is injective on inputs for which \texttt{g$_0$} has a corresponding nonterminal;
\item a proposition \texttt{lift\_nt\_sink} that guarantees that \texttt{sink\_nt} is essentially an inverse of \texttt{lift\_nt};
\item a proposition \texttt{corresponding\_rules} that guarantees that \texttt{g} has a rewrite rule for each rewrite rule \texttt{g$_0$} has (with different type but the same behavior);
\item a proposition \texttt{preimage\_of\_rules} that guarantees that \texttt{g$_0$} has a rewrite rule
for each rewrite rule of \texttt{g} whose nonterminal has a preimage on the \texttt{g$_0$} side.
\end{itemize}

Thanks to this structure, we can abstract from the specifics of how the larger grammar is constructed in concrete proofs and care only about the properties that are required to follow analogous derivations.

To illustrate how we work with this abstraction, we review the proof of the following lemma:
\begin{lstlisting}
private lemma lift_tran {lg : lifted_grammar T}
    {w$_1$ w$_2$ : list (symbol T lg.g$_0$.nt)}
    (hyp : grammar_transforms lg.g$_0$ w$_1$ w$_2$) :
  grammar_transforms lg.g
    (lift_string lg.lift_nt w$_1$)
    (lift_string lg.lift_nt w$_2$)
\end{lstlisting}
We need to show that if \texttt{g$_0$} has a rewrite rule that transforms \texttt{w$_1$} to \texttt{w$_2$}, then \texttt{g} has a rewrite rule that transforms \texttt{lift\_string lg.lift\_nt w$_1$} to \texttt{lift\_string lg.lift\_nt w$_2$}.

We start by deconstructing \texttt{hyp} according to the
\href{https://github.com/madvorak/grammars/blob/99d9f3e6e2fca51b4ab146cd332ff881ff937a29/src/classes/general/basics/definition.lean#L30}
{\texttt{grammar\_transforms}} definition.
To go from \texttt{g$_0$} to \texttt{g},
we first ``lift'' the rewrite rule (i.e., translate its nonterminals in all fields)
that \texttt{g$_0$} used.
We call
\href{https://github.com/madvorak/grammars/blob/99d9f3e6e2fca51b4ab146cd332ff881ff937a29/src/classes/general/basics/lifting.lean#L44}
{\texttt{corresponding\_rules}}
to show that \texttt{g} has such a rule.
Then we use the function
\href{https://github.com/madvorak/grammars/blob/99d9f3e6e2fca51b4ab146cd332ff881ff937a29/src/classes/general/basics/lifting.lean#L17}
{\texttt{lift\_string}}
to lift \texttt{u} and \texttt{v}, which are the parts of the string \texttt{w$_1$} that were not matched by the rule.
We are then left with the proof obligation
\begin{lstlisting}
lift_string lg.lift_nt w$_1$ =
lift_string lg.lift_nt u ++ (lift_rule lg.lift_nt r).input_L
  ++ [symbol.nonterminal (lift_rule lg.lift_nt r).input_N]
  ++ (lift_rule lg.lift_nt r).input_R ++ lift_string lg.lift_nt v
\end{lstlisting}
where
\begin{lstlisting}
w$_1$ =
u ++ r.input_L ++ [symbol.nonterminal r.input_N] ++ r.input_R ++ v
\end{lstlisting}
and with the proof obligation
\begin{lstlisting}
lift_string lg.lift_nt w$_2$ =
lift_string lg.lift_nt u ++ (lift_rule lg.lift_nt r).output_string
  ++ lift_string lg.lift_nt v
\end{lstlisting}
where
\begin{lstlisting}
w$_2$ = u ++ r.output_string ++ v
\end{lstlisting}
These two obligations originate from the two identities in the definition
\href{https://github.com/madvorak/grammars/blob/99d9f3e6e2fca51b4ab146cd332ff881ff937a29/src/classes/general/basics/definition.lean#L30}
{\texttt{grammar\_transforms}}
from Section~\ref{sec:definitions}.
Essentially, we discharge them using the distributivity of
\href{https://github.com/madvorak/grammars/blob/99d9f3e6e2fca51b4ab146cd332ff881ff937a29/src/classes/general/basics/lifting.lean#L17}
{\texttt{lift\_string}}
over the \texttt{++} operation.

The abstraction provided by
\href{https://github.com/madvorak/grammars/blob/99d9f3e6e2fca51b4ab146cd332ff881ff937a29/src/classes/general/basics/lifting.lean#L35}
{\texttt{lifted\_grammar}}
takes care of the vast majority of our proof of the closure of type-0 languages under union. It remains to separately analyze what was the first step of the derivation that \texttt{g} did in the hard direction. We need to exclude all rules that are inherited from \texttt{g$_1$} and \texttt{g$_2$} and perform a case analysis on the two special rules.

The two additional rules of \texttt{g} are context-free. Therefore, if \texttt{g$_1$} and \texttt{g$_2$} have context-free rules only, then all rules of \texttt{g} are context-free as well.
As a consequence, our result about type-0 languages can easily be reused to prove the closure of context-free languages under union:
\begin{lstlisting}
theorem CF_of_CF_u_CF (L$_1$ : language T) (L$_2$ : language T) :
  is_CF L$_1$ $\wedge$ is_CF L$_2\;\;\,\rightarrow\;\;\;$is_CF (L$_1$ + L$_2$)
\end{lstlisting}
Not much Lean code needs to be duplicated
to obtain the result about context-free grammars. We need to write the construction of \texttt{g} again and the main result again. The remaining parts are achieved by reusing lemmas from the proof for general grammars. The main overhead is proving
\begin{lstlisting}
private lemma union_grammar_eq_union_CF_grammar
    {g$_1$ g$_2$ : CF_grammar T} :
  union_grammar (grammar_of_cfg g$_1$) (grammar_of_cfg g$_2$) =
  grammar_of_cfg (union_CF_grammar g$_1$ g$_2$)
\end{lstlisting}
Even though the statement might look complicated, the proof has only five lines, making it one of the shortest tactic-based proofs in our project.

\section{Closure under Reversal}
\label{sec:closure-under-reversal}

In this section, we prove the following theorem:
\begin{lstlisting}
theorem T0_of_reverse_T0 (L : language T) :
  is_T0 L$\;\;\,\rightarrow\;\;\,$is_T0 (reverse_lang L)
\end{lstlisting}

The proof is very easy. Simply speaking, everything gets reversed.
We start with the rewrite rules:
\begin{lstlisting}
private def reversal_grule {N : Type} (r : grule T N) : grule T N :=
grule.mk r.input_R.reverse r.input_N r.input_L.reverse
  r.output_string.reverse
\end{lstlisting}
The constructor
\href{https://github.com/madvorak/grammars/blob/99d9f3e6e2fca51b4ab146cd332ff881ff937a29/src/classes/general/basics/definition.lean#L14}
{\texttt{grule.mk}}
takes arguments in the same order as they are written in the definition:
\begin{itemize}
\item its \texttt{input\_L} is instantiated by \texttt{r.input\_R.reverse};
\item its \texttt{input\_N} is instantiated by \texttt{r.input\_N};
\item its \texttt{input\_R} is instantiated by \texttt{r.input\_L.reverse};
\item its \texttt{output\_string} is instantiated by \texttt{r.output\_string.reverse}.
\end{itemize}

The new grammar is constructed as follows:
\begin{lstlisting}
private def reversal_grammar (g : grammar T) : grammar T :=
grammar.mk g.nt g.initial (list.map reversal_grule g.rules)
\end{lstlisting}

The rest is essentially a repeated application of lemma
\href{https://github.com/madvorak/grammars/blob/99d9f3e6e2fca51b4ab146cd332ff881ff937a29/src/utilities/list_utils.lean#L40}
{\texttt{list.reverse\_append\_append}}, which
is just a repeated application of lemma
\href{https://github.com/leanprover-community/mathlib/blob/7ed4f2cec258c1a1f86cd9a45adf7dc335f42ee1/src/data/list/basic.lean#L610}
{\texttt{list.reverse\_append}}, which
states that reversing two concatenated lists is equivalent to reversing both parts and concatenating them in the opposite order,
and lemma
\href{https://github.com/leanprover-community/mathlib/blob/7ed4f2cec258c1a1f86cd9a45adf7dc335f42ee1/src/data/list/basic.lean#L617}
{\texttt{list.reverse\_reverse}}, which
states that
\href{https://github.com/leanprover-community/lean/blob/154ac72f4ff674bc4486ac611f926a3d6b999f9f/library/init/data/list/basic.lean#L107}
{\texttt{list.reverse}}
is a dual operation.

\section{Closure under Concatenation}
\label{sec:closure-under-concatenation}

In this section, we prove the following theorem:
\begin{lstlisting}
theorem T0_of_T0_c_T0 (L$_1$ : language T) (L$_2$ : language T) :
  is_T0 L$_1$ $\wedge$ is_T0 L$_2\;\;\,\rightarrow\;\;\;$is_T0 (L$_1$ * L$_2$)
\end{lstlisting}
Because the proof is highly technical, we only outline the main idea here.

We first review the classical construction for context-free grammars.
Let $ L_1 \subseteq T^\ast $ be a language generated
by a grammar $ G_1 = ( N_1, T, P_1, S_1 ) $.
Let $ L_2 \subseteq T^\ast $ be a language generated
by a grammar $ G_2 = ( N_2, T, P_2, S_2 ) $.
Without loss of generality, the sets $N_1$ and $N_2$ are disjoint.
We create a new initial symbol $S$ that appears only in the rule $ S \rightarrow S_1 S_2 $.
The new grammar is $ ( N_1 \mskip-2mu\cup\mskip-1mu N_2 \mskip-2mu\cup\mskip-1mu \{ S \},\:
T,\: P_1 \mskip-2mu\cup\mskip-1mu P_2
\mskip-1mu\cup \{ S \rightarrow S_1 S_2 \},\: S\mskip1mu ) $.
This construction works for context-free grammars
because $S_1$ gives rise to a word from $L_1$ and, independently,
$S_2$ gives rise to a word from $L_2$.

For general grammars, the construction above does not work, as
the following counterexample over $ T = \{ a, b \} $ illustrates.
Let the rule sets be
$ P_1 = \{ S_1 \rightarrow S_1 a,\, S_1 \rightarrow \epsilon \} $ and
$ P_2 = \{ S_2 \rightarrow S_2 a,\, S_2 \rightarrow \epsilon,\, a S_2 \rightarrow b \} $.
We obtain $ L_1 = L_2 = \{ a^n \,|\, n \in \mathbb{N}_0 \} $ and
so $ L_1L_2$ is $\{ a^n \,|\, n \in \mathbb{N}_0 \} $ as well.
We can now derive
$ S \Rightarrow S_1 S_2 \Rightarrow S_1 a\mskip1mu S_2 \Rightarrow S_1 b \Rightarrow
b \notin L_1L_2 $ and obtain a contradiction.

We need to avoid matching strings that span across
the boundary of the concatenation. Since the nonterminal sets are disjoint,
the issue arises only with terminals in the left-hand side of rules, which are
not present in context-free grammars.
We provide a solution below.

Let \texttt{g$_1$} and \texttt{g$_2$} generate \texttt{L$_1$} and \texttt{L$_2$} respectively. The nonterminals of our new grammar \texttt{g} consist of
\begin{itemize}
\item the nonterminals of \texttt{g$_1$} including a mark indicating their origin;
\item the nonterminals of \texttt{g$_2$} including a mark indicating their origin;
\item a proxy nonterminal for every terminal from \texttt{T} marked for use by \texttt{g$_1$} only;
\item a proxy nonterminal for every terminal from \texttt{T} marked for use by \texttt{g$_2$} only;
\item one new distinguished nonterminal.
\end{itemize}
The new nonterminal type is encoded by the Lean type
\hbox{\texttt{option$\;$(g$_1$.nt$\,\,\oplus\,\,$g$_2$.nt) $\oplus$ (T$\;\oplus\;$T)}}.
The new distinguished nonterminal becomes the initial symbol of \texttt{g}.

In this way, we ensure that the nonterminals used by \texttt{g} to simulate \texttt{g$_1$} are disjoint from the nonterminals used by \texttt{g} to simulate \texttt{g$_2$}.
There are still real terminals used by both grammars, but \texttt{g} never has these terminals on the left-hand side of a rule, since
the rewrite rules of \texttt{g} consist of
\begin{itemize}
\item the rewrite rules of \texttt{g$_1$} with all nonterminals mapped to the new nonterminal type and all terminals replaced by proxy nonterminals of the first kind;
\item the rewrite rules of \texttt{g$_2$} with all nonterminals mapped to the new nonterminal type and all terminals replaced by proxy nonterminals of the second kind;
\item for every terminal from \texttt{T}, a rule that rewrites the proxy nonterminal of the first kind to the corresponding terminal and a rule that rewrites the proxy nonterminal of the second kind to the corresponding terminal;
\item a special rule that rewrites \texttt{g.initial} to a two-symbol string \texttt{[g$_1$.initial, g$_2$.initial]} wrapped to use the new nonterminal type.
\end{itemize}
Using this construction, we ensure that all rules of \texttt{g} avoid matching strings on the boundary of the concatenation.

Proving that \texttt{g} generates a superset of \texttt{L$_1\!$ * L$_2$} is easy because
we can apply the rewrite rules in the following order, regardless of the languages:
\begin{enumerate}
\item [\,(1)\!] use the special rule to obtain \texttt{[g$_1$.initial, g$_2$.initial]} with the necessary wrapping;
\item [\,(2)\!] generate the string of proxy nonterminals corresponding to the word from \texttt{L$_1$} while \texttt{g$_2$.initial} remains unchanged;
\item [\,(3)\!] replace all proxy nonterminals of the first kind by the corresponding terminals, which results in deriving a word from \texttt{L$_1$} followed by \texttt{g$_2$.initial} as the last symbol;
\item [\,(4)\!] generate the string of proxy nonterminals corresponding to the word from \texttt{L$_2$} while the~first~part of the string remains unchanged;
\item [\,(5)\!] replace all proxy nonterminals of the second kind by the corresponding terminals, which results in deriving a word from \texttt{L$_2$} that follows the word from \texttt{L$_1$} obtained before.
\end{enumerate}
Step (1) is trivial. Steps (2) and (4) are done
by following the derivations by \texttt{g$_1$} and \texttt{g$_2$}, respectively. Steps (3) and (5) are straightforward proofs by induction.

Proving that \texttt{g} generates a subset of \texttt{L$_1$ * L$_2$} is much harder because we do not know in which order the rules of \texttt{g} are applied.
We had to come up with an invariant that relates intermediate strings derived by \texttt{g} to strings that can be derived by \texttt{g$_1$} and \texttt{g$_2$} from their respective initial symbols.

Very roughly speaking, we prove that there are strings \texttt{x} and \texttt{y} for every string \texttt{w} that \texttt{g} can derive, such that the grammar \texttt{g$_1$} can derive \texttt{x}, the grammar \texttt{g$_2$} can derive \texttt{y}, and \hbox{\texttt{x ++ y}}
corresponds to \texttt{w}. As usual, we employ structural induction. \pagebreak[2]
Looking at the last rule \texttt{g} used, we update \texttt{x} or \texttt{y} or neither.
In particular, we want to point out the following declarations
in the formalization:
\begin{itemize}
\item function
\href{https://github.com/madvorak/grammars/blob/99d9f3e6e2fca51b4ab146cd332ff881ff937a29/src/classes/general/closure_properties/concatenation.lean#L98}
{\texttt{nst}} provides the new symbol type which \texttt{g}
operates with;
\item functions
\href{https://github.com/madvorak/grammars/blob/99d9f3e6e2fca51b4ab146cd332ff881ff937a29/src/classes/general/closure_properties/concatenation.lean#L106}
{\texttt{wrap\_symbol$_1$}} and
\href{https://github.com/madvorak/grammars/blob/99d9f3e6e2fca51b4ab146cd332ff881ff937a29/src/classes/general/closure_properties/concatenation.lean#L110}
{\texttt{wrap\_symbol$_2$}}
convert symbols for use by \texttt{g};
\item relation
\href{https://github.com/madvorak/grammars/blob/99d9f3e6e2fca51b4ab146cd332ff881ff937a29/src/classes/general/closure_properties/concatenation.lean#L595}
{\texttt{corresponding\_strings}}
built on top of relation
\href{https://github.com/madvorak/grammars/blob/99d9f3e6e2fca51b4ab146cd332ff881ff937a29/src/classes/general/closure_properties/concatenation.lean#L545}
{\texttt{corresponding\_symbols}}
is used to define how the
strings \texttt{x} and \texttt{y} are precisely related to \texttt{w} after
each step by \texttt{g};
\item lemma
\href{https://github.com/madvorak/grammars/blob/99d9f3e6e2fca51b4ab146cd332ff881ff937a29/src/classes/general/closure_properties/concatenation.lean#L1107}
{\texttt{induction\_step\_for\_lifted\_rule\_from\_g$_1$}}
characterizes the \texttt{x} update;
\item lemma
\href{https://github.com/madvorak/grammars/blob/99d9f3e6e2fca51b4ab146cd332ff881ff937a29/src/classes/general/closure_properties/concatenation.lean#L1911}
{\texttt{induction\_step\_for\_lifted\_rule\_from\_g$_2$}}
characterizes the \texttt{y} update;
\item lemma
\href{https://github.com/madvorak/grammars/blob/99d9f3e6e2fca51b4ab146cd332ff881ff937a29/src/classes/general/closure_properties/concatenation.lean#L2236}
{\texttt{big\_induction}}
states the invariant for proving the hard direction;
\item lemma
\href{https://github.com/madvorak/grammars/blob/99d9f3e6e2fca51b4ab146cd332ff881ff937a29/src/classes/general/closure_properties/concatenation.lean#L2651}
{\texttt{in\_concatenated\_of\_in\_big}}
puts the proof of the hard direction together.
\end{itemize}

Note that the added rules have only one symbol on the left-hand side. Therefore, if the two original
grammars are context-free, our constructed grammar is also context-free.
We thereby obtain, as a bonus, a proof that context-free languages are closed under concatenation.
It is implemented in a similar fashion to the proof that context-free languages are closed under union.

\section{Closure under Kleene Star}
\label{sec:closure-under-kleene-star}

In this section, we prove the following theorem:
\begin{lstlisting}
theorem T0_of_star_T0 (L : language T) :
  is_T0 L$\;\;\,\rightarrow\;\;\,$is_T0 L.star
\end{lstlisting}

This is usually demonstrated by a hand-waving argument about
a two-tape nondeterministic Turing machine.
The language to be iterated is given by a
single-tape (nondeterministic) Turing machine.
The new machine scans the input on the first tape,
copying it onto the second tape as it progresses,
and nondeterministically chooses where the first word ends.
Next, the original machine is simulated on the second tape.
If the simulated machine accepts the word on the second tape,
the process is repeated with the current position of the first head
instead of returning to the beginning of the input.
Finally, when the first head reaches the end of the input,
the second tape contains a suffix of the first tape.
The original machine is simulated once more on the second tape.
If it accepts, the new machine accepts.

Unfortunately, we did not find any proof based on grammars in the literature.
Therefore, we had to invent our own construction.
In Section~\ref{sec:star-sketch}, we present the construction and the idea
underlying its correctness using traditional mathematical notation. In
Section~\ref{sec:star-formalization}, we comment on its formalization.

\subsection{Proof Sketch}
\label{sec:star-sketch}

Let $L \subseteq T^\ast$ be a language generated by the grammar $G = (N, T, P, S)$.
We construct a grammar $G_\ast = (N_\ast, T, P_\ast, Z)$
to generate the language $L^\ast$. The new nonterminal set
\[ N_\ast = N \cup \{Z, \#, R\} \]
expands $N$ with three additional
nonterminals:\ a new starting symbol ($Z$), a delimiter ($\#$),
and a marker for final rewriting ($R$).
The new set of rules is
\[ P_\ast = P \cup \{ Z \rightarrow ZS\#,\, Z \rightarrow R\#,\, R\# \rightarrow R,\, R\# \rightarrow \epsilon \} \cup \{ Rt \rightarrow tR ~|~ t \in T \} \]
Intuitively, $\#$ builds compartments that isolate the words from
the language $L$,
and then $R$ acts as a cleaner that traverses the string from
beginning to end and removes the compartment delimiters $\#$,
thereby ensuring that only terminals are present to the left of $R$.

To see how $G_\ast$ works, consider the following grammar over $ T = \{ a, b \} $. Let $ N = \{ S \} $ and $ P = \{ S \rightarrow a S b,\, S \rightarrow \epsilon \} $.
The set of rules becomes
\[ P_\ast = \{ S \rightarrow aSb,\, S \rightarrow \epsilon,\, Z \rightarrow ZS\#,\, Z \rightarrow R\#,\, R\# \rightarrow R,\, R\# \rightarrow \epsilon,\, Ra \rightarrow aR,\, Rb \rightarrow bR \} \]
The following is an example of $G_\ast$ derivation:
\begin{multline*}
Z \Rightarrow ZS\# \Rightarrow ZS\#S\# \Rightarrow ZaSb\#S\# \Rightarrow
ZaaSbb\#S\# \Rightarrow
ZS\#aaSbb\#S\# \Rightarrow \\
ZaSb\#aaSbb\#S\# \Rightarrow
ZaSb\#aaaSbbb\#S\# \Rightarrow
ZaSb\#aaabbb\#S\# \Rightarrow \\
R\#aSb\#aaabbb\#S\# \Rightarrow
RaSb\#aaabbb\#S\# \Rightarrow
aRSb\#aaabbb\#S\# \Rightarrow \\
aRb\#aaabbb\#S\# \Rightarrow
abR\#aaabbb\#S\# \Rightarrow
abRaaabbb\#S\# \Rightarrow \\
abaRaabbb\#S\# \Rightarrow
abaaRabbb\#S\# \Rightarrow
abaaaRbbb\#S\# \Rightarrow \\
abaaabRbb\#S\# \Rightarrow
abaaabRbb\#aSb\# \Rightarrow
abaaabbRb\#aSb\# \Rightarrow \\
abaaabbRb\#ab\# \Rightarrow
abaaabbbR\#ab\# \Rightarrow
abaaabbbRab\# \Rightarrow \\
abaaabbbaRb\# \Rightarrow
abaaabbbabR\# \Rightarrow
abaaabbbab\hfill
\end{multline*}

\begin{lemma} \label{lemma-easy}
Let $w_1, w_2, \dots, w_n \in L$. Then $G_\ast$ can derive $Z w_1 \# w_2 \# \dots w_n \#$.
\end{lemma}

\begin{proof}
By induction on $n$. The base case $Z \Rightarrow^\ast Z$ is trivial.

Now assume $\, Z \Rightarrow^\ast Z w_1 \# w_2 \# \dots w_n \# \,$ and $\: S \Rightarrow^\ast w_{n+1} \,$.
We start with the rule $Z \rightarrow ZS\#$. We observe $ZS\# \Rightarrow^\ast Z w_1 \# w_2 \# \dots w_n \# S\# \Rightarrow^\ast Z w_1 \# w_2 \# \dots w_n \# w_{n+1} \#$.
By~transitivity, we obtain $Z \Rightarrow^\ast Z w_1 \# w_2 \# \dots w_n \# w_{n+1} \#$.
\end{proof}
From now on, let $\inde m$ denote the set of $m$ natural numbers $\{ 1, 2, \dots, m \}$.

\begin{lemma} \label{lemma-hard}
If $\alpha \in (T \cup N)^\ast$ can be derived by $G_\ast$, then one of these conditions holds:
\begin{enumerate}
	\item $\exists\mskip1mu x_1, x_2, \dots, x_m \in (T \cup N)^\ast \left( \mskip1mu\forall\mskip1mu i \in\! \inde m \left( S \Rightarrow^\ast\! x_i \right) \aand \alpha = Z x_1\# x_2\# \dots x_m \# \right) ${\rm;}
	\item $\exists\mskip1mu x_1, x_2, \dots, x_m \in (T \cup N)^\ast \left( \mskip1mu\forall\mskip1mu i \in\! \inde m \left( S \Rightarrow^\ast\! x_i \right) \aand \alpha = R\# x_1\# x_2\# \dots x_m \# \right) ${\rm;}
	\item $\exists\mskip1mu w_1, w_2, \dots, w_n \in L \;( \exists \beta \in T^\ast \;( \exists \gamma, x_1, x_2, \dots, x_m \in (T \cup N)^\ast $ \\
	$\hbox{}\qquad \left( S \Rightarrow^\ast\! \beta \gamma \aand \forall\mskip1mu i \in\! \inde m \left( S \Rightarrow^\ast\! x_i \right) \aand \alpha = w_1 w_2 \dots w_n \,\beta R \gamma \# x_1\# x_2\# \dots x_m \# \right) ) ) ${\rm;}
	\item $\alpha \in L^\ast${\rm;}
	\item $\exists\mskip1mu \sigma \in (T \cup N)^\ast \left( \alpha = \sigma R \right) ${\rm;}
	\item $\exists\mskip1mu \omega \in (T \cup N \cup \{\#\})^\ast \left( \alpha = \omega\# \right) ${\rm.}
\end{enumerate}
\end{lemma}

In the example above, case 1 arises when $ \alpha = ZaaSbb\#S\# $. We can check that $m=2$, $x_1 = aaSbb$, $x_2 = S$, and condition 1 holds.

In the example above, case 2 arises when $ \alpha = R\#aSb\#aaabbb\#S\# $. We can check that $m=3$, $x_1 = aSb$, $x_2 = aaabbb$, $x_3 = S$, and condition 2 holds.

In the example above, case 3 arises when $ \alpha = abaaabRbb\#aSb\# $. We can check that $n=1$, $m=1$, $w_1 = ab$, $\beta = aaab$, $\gamma = bb$, $x_1 = aSb$, and condition 3 holds.

Case 4 arises only at the end of a successful computation, which is $\alpha = abaaabbbab$ in the example above.

The remaining two cases do not arise in the example above because they describe an unsuccessful computation (like taking a one-way street ending in a blind alley).

Case 5 arises if the rule $R\# \rightarrow R$ is used in the final position (where $R\# \rightarrow \epsilon$ should be used instead). The nonterminal $R$ in the final position prevents the derivation from terminating.

Case 6 arises if the rule $R\# \rightarrow \epsilon$ is used too early (that is, anywhere but the final $\#$ position). The nonterminal $\#$ in the final position during the absence of $R$ and $Z$ in $\alpha$ prevents the derivation from terminating.

\begin{proof}
By induction on $G_\ast$ derivation steps. The base case $\alpha = Z$ satisfies condition 1 by setting $m = 0$.

Now assume $Z \Rightarrow^\ast \alpha \Rightarrow \alpha'$ and proceed by case analysis on the conditions.\pagebreak[2]
\begin{enumerate}
	\item $\exists\mskip1mu x_1, x_2, \dots, x_m \in (T \cup N)^\ast \left( \mskip1mu\forall\mskip1mu i \in\! \inde m \left( S \Rightarrow^\ast\! x_i \right) \aand \alpha = Z x_1\# x_2\# \dots x_m \# \right) $:
	\begin{itemize}
		\item If $\alpha \Rightarrow \alpha'$ used a rule from $P$, it could be applied only in some $x_i$. Hence $S \Rightarrow^\ast x_i \Rightarrow x_i'$, so the same condition holds after replacing $x_i$ by $x_i'$.
		\item If $\alpha \Rightarrow \alpha'$ used the rule $Z \rightarrow ZS\#$, it was applied at the beginning of $\alpha$. Therefore, we set $m' := m + 1$, we set $x_1' := S$, and we increase all indices by one, that is, $x_2' := x_1,\; x_3' := x_2,\; \dots,\, x_{m'}' := x_m$. The same condition holds.
		\item If $\alpha \Rightarrow \alpha'$ used the rule $Z \rightarrow R\#$, we keep all variables the same and condition~2 holds.
		\item The rules $R\# \rightarrow R$, $R\# \rightarrow \epsilon$, and $Rt \rightarrow tR$ are not applicable (since $\alpha$ does not contain $R$).
	\end{itemize}
	\item $\exists\mskip1mu x_1, x_2, \dots, x_m \in (T \cup N)^\ast \left( \mskip1mu\forall\mskip1mu i \in\! \inde m \left( S \Rightarrow^\ast\! x_i \right) \aand \alpha = R\# x_1\# x_2\# \dots x_m \# \right) $:
	\begin{itemize}
		\item If $\alpha \Rightarrow \alpha'$ used a rule from $P$, it could be applied only in some $x_i$. Hence $S \Rightarrow^\ast x_i \Rightarrow x_i'$, so the same condition holds after replacing $x_i$ by $x_i'$.
		\item The rules $Z \rightarrow ZS\#$ and $Z \rightarrow R\#$ are not applicable (since $\alpha$ does not contain $Z$).
		\item If $\alpha \Rightarrow \alpha'$ used the rule $R\# \rightarrow R$, it was applied at the beginning of $\alpha$. If $m = 0$, condition~5 holds (a dead end). Otherwise, we set $m' := m - 1 \ge 0$ and $\gamma := x_1$, and we decrease all indices by one, that is, $x_1' := x_2,\; x_2' := x_3,\; \dots,\; x_{m'}' := x_m$. Since there is nothing before the nonterminal $R$, we set $n := 0$ and $\beta := \epsilon$. Now, condition~3 holds.
		\item If $\alpha \Rightarrow \alpha'$ used the rule $R\# \rightarrow \epsilon$ then: if $m=0$, we obtain the empty word (which belongs to $L^\ast$, satisfying condition~4); if $m>0$, condition~6 holds (because $\#$ remained at the end of $\alpha'$; at the same time $R$ disappeared, and $Z$ did not appear).
		\item The rule $Rt \rightarrow tR$ is not applicable (the only $R$ in $\alpha$ is immediately followed by $\#$).
	\end{itemize}
	\item $\exists\mskip1mu w_1, w_2, \dots, w_n \in L \;( \exists \beta \in T^\ast \;( \exists \gamma, x_1, x_2, \dots, x_m \in (T \cup N)^\ast $ \\
	$\hbox{}\qquad \left(S \Rightarrow^\ast\! \beta \gamma \aand \forall\mskip1mu i \in\! \inde m \left( S \Rightarrow^\ast\! x_i \right) \aand \alpha = w_1 w_2 \dots w_n \,\beta R \gamma \# x_1\# x_2\# \dots x_m \# \right) ) ) $:
	\begin{itemize}
		\item If $\alpha \Rightarrow \alpha'$ used a rule from $P$, it could be applied in $\gamma$ or in some $x_i$. In the first case, $\gamma \Rightarrow \gamma'$ implies $\beta\gamma \Rightarrow \beta\gamma'$, hence $S \Rightarrow^\ast \beta\gamma \Rightarrow \beta'\gamma'$. In the remaining cases, we observe $S \Rightarrow^\ast x_i \Rightarrow x_i'$ as we did at the beginning of our case analysis. As a result, the same condition still holds.
		\item The rules $Z \rightarrow ZS\#$ and $Z \rightarrow R\#$ are not applicable ($\alpha$ does not contain $Z$).
		\item If $\alpha \Rightarrow \alpha'$ used the rule $R\# \rightarrow R$, then $\gamma$ must have been empty. If $m = 0$, condition~5 holds (a dead end). Otherwise, we set $n' := n + 1$, $w_{n'} := \beta$, $\beta' := \epsilon$, $\gamma' := x_1$, and $m' := m - 1$, and we decrease the indices of $x_i$ by one, that is, $x_1' := x_2,\; x_2' := x_3,\; \dots,\; x_{m'}' := x_m$. Since $w_{n'} = \beta = \beta \gamma \in T^\ast$ and $S \Rightarrow^\ast \beta \gamma$, we have $w_{n'} \in L$. The same condition holds.
		\item If $\alpha \Rightarrow \alpha'$ used the rule $R\# \rightarrow \epsilon$, then $\gamma$ must have been empty. If~$m=0$, we get $\alpha = w_1 w_2 \dots w_n \beta$ and $\beta \in L$; hence condition~4, $\alpha' \in L^\ast$, is satisfied. If $m>0$, condition~6 now holds (because $\#$ remained at the end of $\alpha'$; at the same time $R$ disappeared, and $Z$ did not appear).
		\item If $\alpha \Rightarrow \alpha'$ used a rule of the form $Rt \rightarrow tR$ ($t \in T$), we have $\delta \in (T \cup N)^\ast$ such that $\gamma = t \delta$. We put $\beta' := \beta t$ and $\gamma' := \delta$. Since $\beta \gamma = \beta t \delta = \beta' \gamma'$, the same condition holds.
	\end{itemize}
	\item $\alpha \in L^\ast$:
	\begin{itemize}
		\item No rule is applicable (since $\alpha$ contains only terminals). The step $\alpha \Rightarrow \alpha'$ cannot have happened.
	\end{itemize}
	\item $\exists\mskip1mu \sigma \in (T \cup N)^\ast \left( \alpha = \sigma R \right) $:
	\begin{itemize}
		\item No matter which rule was applied, it happened within $\sigma$. No rule could match the final $R$. The same condition holds for $\alpha' = \sigma' R$.
	\end{itemize}
	\item $\exists\mskip1mu \omega \in (T \cup N \cup \{\#\})^\ast \left( \alpha = \omega\# \right) $:
	\begin{itemize}
		\item If $\alpha \Rightarrow \alpha'$ used a rule from $P$, the same condition still holds because $\#$ is not on the left-hand side of any rule from $P$ and neither $Z$ nor $R$ is on the right-hand side of any rule from $P$.
		\item The rules $Z \rightarrow ZS\#$, $Z \rightarrow R\#$, $R\# \rightarrow R$, $R\# \rightarrow \epsilon$, and $Rt \rightarrow tR$ are not applicable (since $\alpha$ contains neither $Z$ nor $R$).  \qedhere
	\end{itemize}
\end{enumerate}
\end{proof}

\begin{theorem}
The class of type-0 languages is closed under the Kleene star.
\end{theorem}

\begin{proof}
We need to show that the language of $G_\ast$ equals $L^\ast$.
We prove two inclusions.

For ``$\supseteq L^\ast$'', we use Lemma \ref{lemma-easy}. If $w \in L^\ast$, there exist words $w_1, w_2, \dots, w_n \in L$ such that $w_1 w_2 \dots w_n = w$. We see $Z w_1 \# w_2 \# \dots w_n \# \Rightarrow R\# w_1 \# w_2 \# \dots w_n \#$. Since all words $w_i$ are made of terminals only, by repeated application of $R\# \rightarrow R$ and $Rt \rightarrow tR$ (for all $t \in T$) we get $R\# w_1 \# w_2 \# \dots w_n \# \Rightarrow^\ast w_1  w_2 \dots w_n R \# $. Finally, $w_1  w_2 \dots w_n R \# \Rightarrow w_1 w_2 \dots w_n$ is obtained by the rule $R\# \rightarrow \epsilon$. We conclude that $G_\ast$ generates $w$.

For ``$\subseteq L^\ast$'', we use Lemma \ref{lemma-hard} and observe that if $G_\ast$ generates $\alpha \in T^\ast$, then $\alpha \in L^\ast$ because all the remaining cases require $\alpha$ to contain a nonterminal.
\end{proof}

\subsection{Formalization}
\label{sec:star-formalization}

The formalization closely follows the proof sketch. The main difference between
the two is that where the proof sketch states that an expression belongs to a
set, the formalization specifies a type for a term and sometimes a condition that
further restricts the term's values.

Lemma \ref{lemma-easy} is implemented by lemma
\href{https://github.com/madvorak/grammars/blob/99d9f3e6e2fca51b4ab146cd332ff881ff937a29/src/classes/general/closure_properties/star.lean#L84}
{\texttt{short\_induction}}, which
takes $w$ in reverse order for technical reasons. Its proof uses the
\href{https://github.com/madvorak/grammars/blob/99d9f3e6e2fca51b4ab146cd332ff881ff937a29/src/classes/general/basics/lifting.lean#L35}
{\texttt{lifted\_grammar}}
approach outlined in Section~\ref{sec:closure-under-union}. The part $R\# w_1 \# w_2 \# \dots w_n \# \Rightarrow^\ast w_1  w_2 \dots w_n R \# $ is implemented by lemma
\href{https://github.com/madvorak/grammars/blob/99d9f3e6e2fca51b4ab146cd332ff881ff937a29/src/classes/general/closure_properties/star.lean#L306}
{\texttt{terminal\_scan\_ind}}, which
employs a nested induction to pass $R$ to the right. The final step of the easy direction is performed inside the theorem
\href{https://github.com/madvorak/grammars/blob/99d9f3e6e2fca51b4ab146cd332ff881ff937a29/src/classes/general/closure_properties/star.lean#L3938}
{\texttt{T0\_of\_star\_T0}} itself.

Lemma \ref{lemma-hard} is implemented by lemma
\href{https://github.com/madvorak/grammars/blob/99d9f3e6e2fca51b4ab146cd332ff881ff937a29/src/classes/general/closure_properties/star.lean#L3870}
{\texttt{star\_induction}}, whose
formal proof spans over 3000 lines. The base case is discharged immediately. For the induction step, we developed six lemmas
\href{https://github.com/madvorak/grammars/blob/99d9f3e6e2fca51b4ab146cd332ff881ff937a29/src/classes/general/closure_properties/star.lean#L1171}
{\texttt{star\_case\_1}} to
\href{https://github.com/madvorak/grammars/blob/99d9f3e6e2fca51b4ab146cd332ff881ff937a29/src/classes/general/closure_properties/star.lean#L3647}
{\texttt{star\_case\_6}}
distinguished by which of the six conditions $\alpha$ satisfies. In each of them, except for
\href{https://github.com/madvorak/grammars/blob/99d9f3e6e2fca51b4ab146cd332ff881ff937a29/src/classes/general/closure_properties/star.lean#L3359}
{\texttt{star\_case\_4}}, which
took only four lines to prove, we perform a case analysis on which rule was used for the $\alpha \Rightarrow \alpha'$ transition.

For each case, unless a short
ex-falso-quodlibet proof suffices, we need
to narrow down where in $\alpha$ the rule could be applied. This analysis is challenging for the rules that were inherited from the original grammar. Consider
\href{https://github.com/madvorak/grammars/blob/99d9f3e6e2fca51b4ab146cd332ff881ff937a29/src/classes/general/closure_properties/star.lean#L1067}
{\texttt{case\_1\_match\_rule}},
where the informal argument literally says:
``If $\alpha \Rightarrow \alpha'$ used a rule from $P$, it could be applied only in some $x_i$.''

It turns out that this deduction is so complicated that it was worth creating an auxiliary lemma
\href{https://github.com/madvorak/grammars/blob/99d9f3e6e2fca51b4ab146cd332ff881ff937a29/src/classes/general/closure_properties/star.lean#L728}
{\texttt{cases\_1\_and\_2\_and\_3a\_match\_aux}}
to detach the head $Z$ from $\alpha$ and perform the analysis on $x_1\# x_2\# \dots x_m \#$ in order to make the proof easier.
As a useful side effect, the auxiliary lemma becomes applicable to similar situations in
\href{https://github.com/madvorak/grammars/blob/99d9f3e6e2fca51b4ab146cd332ff881ff937a29/src/classes/general/closure_properties/star.lean#L1638}
{\texttt{star\_case\_2}} and
\href{https://github.com/madvorak/grammars/blob/99d9f3e6e2fca51b4ab146cd332ff881ff937a29/src/classes/general/closure_properties/star.lean#L2925}
{\texttt{star\_case\_3}}, as shown in
\href{https://github.com/madvorak/grammars/blob/99d9f3e6e2fca51b4ab146cd332ff881ff937a29/src/classes/general/closure_properties/star.lean#L1510}
{\texttt{case\_2\_match\_rule}} and
\href{https://github.com/madvorak/grammars/blob/99d9f3e6e2fca51b4ab146cd332ff881ff937a29/src/classes/general/closure_properties/star.lean#L2267}
{\texttt{case\_3\_match\_rule}}, where
more adaptations are needed but the same core argument is used.

From a formal point of view, we abused the symbol `${\dots}$' in the proof sketch. Replacing it by a formal statement usually leads to
\href{https://github.com/leanprover-community/lean/blob/154ac72f4ff674bc4486ac611f926a3d6b999f9f/library/init/data/list/basic.lean#L128}
{\texttt{list.join}} of
\href{https://github.com/leanprover-community/lean/blob/154ac72f4ff674bc4486ac611f926a3d6b999f9f/library/init/data/list/basic.lean#L110}
{\texttt{list.map}}
of something. For example, compare case 1 in the proof sketch
\[\exists\mskip1mu x_1, x_2, \dots, x_m \in (T \cup N)^\ast \left( \mskip1mu\forall\mskip1mu i \in\! \inde m \left( S \Rightarrow^\ast\! x_i \right) \aand \alpha = Z x_1\# x_2\# \dots x_m \# \right)\]
to its formal counterpart:
\begin{lstlisting}
$\exists$ x : list (list (symbol T g.nt)),
  ($\forall$ x$_i$ $\in$ x, grammar_derives g [symbol.nonterminal g.initial] x$_i$) $\wedge$
  ($\alpha$ = [Z] ++ list.join
    (list.map (++ [H]) (list.map (list.map wrap_sym) x)))
\end{lstlisting}
The nonterminal \# is represented by the letter \texttt{H} in the code. Notice how easy it is to write the quantification $\exists\mskip1mu x_1, x_2, \dots, x_m \in (T \cup N)^\ast$ in Lean. The part $\forall i \in\! \inde m \left(S \Rightarrow^\ast x_i\right)$ is also elegant. However, the expression $Z x_1\# x_2\# \dots x_m \#$ leads to a fairly complicated Lean term.

Because many lemmas need to work with expressions like the above, it is important to master how to manipulate terms that combine
\href{https://github.com/leanprover-community/lean/blob/154ac72f4ff674bc4486ac611f926a3d6b999f9f/library/init/data/list/basic.lean#L128}
{\texttt{list.join}}
with other functions. For example, the following lemma is useful:
\begin{lstlisting}
lemma append_join_append {s : list $\alpha$} (L : list (list $\alpha$)) :
  s ++ (list.map ($\lambda$ l, l ++ s) L).join =
       (list.map ($\lambda$ l, s ++ l) L).join ++ s
\end{lstlisting}
This lemma allows us to move the parentheses in $ s\mskip1mu (l_1 s) (l_2 s) \dots (l_n s) $ to get $ (s\mskip1mu l_1) (s\mskip1mu l_2) \dots (s\mskip1mu l_n\mskip-1mu)\mskip1mu s $ and vice versa.

Working with expressions such as $Z x_1\# x_2\# \dots x_m \#$ is tedious in
Lean. We see this, however, not as a weakness of Lean but rather as an
indication that the `${\dots}$' notation is highly informal. Mathematical
expressions with `${\dots}$' tend to be ambiguous and require the reader's
cooperation to make sense of them. In the absence of support for `${\dots}$'
in the proof assistant \cite{flexdots},
it is natural that formalizing such expressions leads to verbose code.

In contrast to concatenation,
the above proof cannot be reused to establish the closure of context-free languages
under the Kleene star because
our construction adds rules with two symbols on their left-hand side.
However, there exists an easier construction for context-free languages
that could be formalized separately if desired.

\section{Related Work}
\label{sec:related-work}

To our knowledge, no one has formalized general grammars before. Context-free
grammars were formalized by Carlson et al.\ \cite{mizarCF} using Mizar, by
Minamide \cite{isabelleCF} using Isabelle/HOL, by Barthwal and Norrish
\cite{holCF} using HOL4, by Firsov and Uustalu \cite{agdaCF} using Agda, and by
Ramos \cite{coqCF} using Coq.

Finite automata have often been subjected to verification. In particular,
Thompson and Dillies \cite{mathlib-mathlib-2020} formalized finite automata,
which recognize regular languages, using Lean.
Thomson \cite{mathlib-mathlib-2020} also formalized regular expressions,
which recognize regular languages as well.

There is ample verification work also for other models of computation:
\begin{itemize}
\item Turing machines were formalized using Mizar \cite{mizarTM},
  Matita \cite{matitaTM}, Isabelle/HOL \cite{isabelleTM}, Lean
  \cite{CarneiroPRF}, Coq \cite{KunzeTuring},
  and recently again Isabelle/HOL \cite{Cook_Levin-AFP}. Of
  these, the most impressive development is probably the last one, by Balbach.
  It uses multi-tape Turing machines and culminates with a proof of the
  Cook--Levin theorem, which states that SAT is \textbf{NP}-complete.

\item The $\lambda$-calculus was formalized by Norrish \cite{NorrishRF} using HOL4
  and later by Forster, Kunze, and their colleagues
\cite{ForsterLambdaCoq,KunzeSmallStep,KunzeTimeBounds,ForsterTimeSpace,ForsterTimeInvariance,KunzeCook} using Coq. The latter group of authors proposed an
  untyped call-by-value $\lambda$-calculus as a convenient basis for
  computability and complexity theory because it naturally supports
  compositionality.

\item The partial recursive functions were formalized by
  Norrish \cite{NorrishRF} using HOL4
  and by Carneiro \cite{CarneiroPRF} using Lean.

\item Random access machines were formalized by Coen \cite{LindaCalculus} using 
  Coq.
\end{itemize}

\section{Conclusion}
\label{sec:conclusion}

We defined general grammars in Lean and used them to establish closure properties
of recursively enumerable or type-0 languages. We found that closure under union and reversal were
straightforward to formally prove, but had to invest considerable effort to
prove closure under concatenation and the Kleene star. Despite the tedium of some of
the proofs, we believe that grammars are probably a more convenient formalism
than Turing machines for showing closure properties.
On the other hand, since grammars do not define any of the important complexity
classes (such~as~\textbf{P}), formalization of Turing machines and other
computational models is needed to further develop the formal theory of
computer science.

As future work, results about context-sensitive, context-free,
and regular grammars could be incorporated into our library.
A comprehensive Lean library encompassing the entire
Chomsky hierarchy would be valuable. We already have
some results about context-free grammars, and the \texttt{mathlib} results about
regular languages could be connected to our library. As a more ambitious goal,
we might attempt to prove the equivalence between general grammars and Turing
machines.

\bibliography{bib}

\end{document}